\documentclass[11pt]{amsart}

\usepackage{amsmath}
\usepackage{amsthm}
\usepackage{amsfonts}
\usepackage{amssymb}
\usepackage{enumerate}
\usepackage{asymptote}

\usepackage{graphics}
\usepackage{graphicx}
\usepackage{epsfig}
\usepackage{color}
\usepackage{verbatim}
\usepackage[ruled,vlined]{algorithm2e}
\include{pythonlisting}
\newtheorem{theorem}{Theorem}
\newtheorem{lemma}[theorem]{Lemma}

\newtheorem{claim}[theorem]{Claim}

\newcommand{\disc}{\mathsf{disc}}
\newcommand{\proj}{\mathsf{Proj}}

\usepackage[backend=bibtex8,style=alphabetic,maxnames=99,sorting=anyt]{biblatex}
\begin{document}

\title{Algorithmic Discrepancy Minimization}
\date{May 12, 2019}
\author{Jonathan Liu and Michael Whitmeyer}

\maketitle

\section{Abstract}

This report will be a literature review on a result in algorithmic discrepancy theory. We will begin by providing a quick overview on discrepancy theory and some major results in the field, and then focus on an important result by Shachar Lovett and Raghu Meka in \cite{Lovett}. We restate the main algorithm and ideas of the paper, and rewrite proofs for some of the major results in the paper.

\section{Introduction}
The discrepancy problem is as follows: given a finite family of finite sets of points, our goal is to color the underlying points (contained in the union of all the sets in the family) red and blue, such that each set has a roughly equal number of red and blue points. Formally, it is described as follows:
Given a universe $[n] = \{1, ... ,n\}$ and a collection of sets in the universe $S = \{S_1,...,S_m \subseteq [n]\}$, We wish to find an assignment $\chi: [n] \rightarrow \{-1, 1\}$
such that $\disc(\chi)$ is minimized, where $\disc(\chi)$ is defined as
\[ \disc(\chi) = \max_{S_i \in S} \left|\sum_{i \in S_i} \chi(i)\right|. \]
Perhaps two of the most major results in Discrepancy Theory came in the 1980's, when two papers published proofs of the existence of assignments with surprisingly strong lower bounds. First, in 1981, J\'ozsef Beck and Tibor Fiala showed that given an upper limit on the number of sets that each point is included in, we can find an assignment with discrepancy linear in that limit.

\begin{theorem}[\cite{Beck-Fiala}]
We start with the assumption that for each $x \in [n]$, it appears in at most $t$ sets. More formally, we have the constraint that $\forall i \in [n]$,
$$|\{j; i\in S_j\}| \leq t$$
Then, one can find an assignment $\chi$ such that $\disc(\chi) \leq 2t-1$.
\end{theorem}

We provide a proof of the Beck-Fiala theorem in the appendix in \ref{Beck-Fiala}, using only arguments from linear algebra.

The other groundbreaking result in Discrepancy Theory is called Spencer's Six Standard Deviations and is given here:
\begin{theorem}[\cite{Spencer}]\label{Spencer}
Given any system of $n$ sets on a universe of $n$ points, there exists a coloring $\chi$ such that $\disc(\chi) \leq 6\sqrt{n}$.
\end{theorem}

Both of these results remain cornerstones of Discrepancy Theory. Yet, despite their significance, they were both proven using nonconstructive methods, so we had no way to achieve them algorithmically. For some time, it was even conjectured that no algorithm could be provided. This question remained open until \cite{Bansal}, where Nikhil Bansal provides a constructive randomized algorithm for discrepancy minimization based on an SDP relaxation. Later, Lovett and Meka propose a new constructive algorithm using only linear algebra \cite{Lovett}. We will be focusing on this paper. 

\section{Overview}
The paper provides a constructive algorithm for minimizing discrepancy, and uses it to prove that their bounds match the bounds given by the previously mentioned theorems. First, they demonstrate a result matching Spencer's Six Standard Deviations.
\begin{theorem}[\cite{Lovett}]\label{result}
For any system of $m$ sets on a universe of $n$ points, there exists a randomized algorithm that, in polynomial time and with at least $1/2$ chance, computes a coloring $\chi: [n] \to \{-1, 1\}$ such that $\disc(\chi) < K\sqrt{n\log_2(m/n)}$ for some universal constant $K$.
\end{theorem}

We note that for $m = n$, as in the case of Theorem \ref{Spencer}, we reach the same asymptotic bound as Spencer provided. 
\cite{Lovett} then provides a similar result for the "Beck-Fiala case" where the occurrence of each point is upper bounded

\begin{theorem}[\cite{Lovett}]
For a system of $m$ sets on a universe of $n$ points where each point is contained in at most $t$ sets, there exists a randomized algorithm that, in polynomial time and with at least $1/2$ chance, computes a coloring $\chi: [n] \to \{-1, 1\}$ such that $\disc(\chi) < K\sqrt{t}\log n$ for some universal constant $K$.
\end{theorem}

In this review we will focus on Theorem 3. The main idea behind the algorithm will be to first create a \textit{partial coloring}. Given a set-system $(V,S)$, where $V=\{1,..., n\}$ and $|S| = m$, we assume that $m \geq n$ (the other case can easily be reduced to this case by adding some empty sets to $S$). We call this partial coloring $\chi : V \rightarrow [-1,1]$ such that:
\begin{enumerate}
    \item For all $S_i \in S, |\chi(S_i)| = O(\sqrt{n\log(m/n)})$
    \item $|\{i : |\chi(i)| = 1\}| \geq cn$ for a constant $c>0$
    
\end{enumerate}
The idea is that if we are provided a good algorithm for finding a partial coloring, we can repeatedly apply this algorithm on the variables not yet "colored" by this partial coloring, while holding the colored ones constant. This will eventually converge to a full coloring and total discrepancy, as the number of points colored follows a geometric series with ratio $\sqrt{1-c}$, and the discrepancy can be bounded by $O(\sqrt{n\log(m/n)})$.

\section{Achieving a Partial Coloring}

The first and most important step is actually achieving a good partial coloring. We start with a convenient construction: We let $v_1,...,v_m \in \mathbb{R}^n$ be the indicator vectors for each of our subsets $S_1,...,S_m$ respectively. Then, the discrepancy of our collection $S$ can be very easily described as 
$$\disc(S) = \max_{i \in [m]} |\langle \chi, v_i \rangle|$$
\begin{theorem}[Main Partial Coloring Result]\label{Coloring}
Let $v_1,...,v_m \in \mathbb{R}^n$ be vectors, and $x_0 \in [-1,1]^n$ be a "starting point". Further let $c_1,...,c_m \geq 0$ be thresholds such that $\sum_{j=1}^m\exp(-c_j^2/16) \leq n/16$. Then let $\delta>0$ be an approximation parameter. Then there exists an efficient randomized algorithm which with probability $\geq 0.1$ finds $x \in [-1,1]^n$ such that
\begin{enumerate}
    \item discrepancy constraints: \quad \quad  $|\langle x-x_0,v_j\rangle| \leq c_j \|v_j\|_2$ 
    \item variable constraints: \quad \quad $|x_i| \geq 1-\delta$ for at least $n/2$ indices $i \in [n]$
\end{enumerate}
Moreover, the algorithm runs in time $O((m+n)^3\delta^{-2}\log (nm/\delta))$
\end{theorem}

The reason for the constraint on the $c_j$'s will become apparent later, but for now we note that the smaller the $c_j$ are, the stronger the theorem is.  In other words, we want them to be small, but they can't be too small otherwise the theorem won't hold, hence the constraint. We also note that we can increase the probability of success by simply running the algorithm multiple times over.

\begin{center}
    4.1. The Algorithm.
\end{center}
We begin with a general idea, before going into the details of the algorithm. We also assume without changing the problem that the $v_i$'s have all been normalized (we can simply adjust our $c_j$'s to account for this): $\|v_i\|_2 = 1, \forall i$. Consider the following polytope, which describes the legal values that $x\in \mathbb{R}^n$ can take on:
$$\mathcal{P} = \{x\in \mathbb{R}^n: |x_i|\leq 1 \forall i \in [n], |\langle x-x_0,v_j\rangle| \leq c_j\}.$$

Then the above theorem says we can find an $x\in \mathbb{R}^n$ such that at least $n/2$ of the variable constraints are satisfied with (virtually) no slack, and it works with good probability as long as we have $\sum \exp(-c_j^2) << n$. The idea is to take very small, discrete, Gaussian steps (called Brownian motion) starting from $x_0$. We intuitively want to use these steps to find such an $x$ that is as far away from the origin ($x_0$) as possible, as this implies that more of the constraints are satisfied with no slack. 

We are now ready to present the constructive algorithm that serves as a proof of Theorem \ref{Coloring}. Let $\gamma >0$ be a small step size such that $\delta=O(\gamma \sqrt{\log(nm/\gamma})$. The correctness of the algorithm will not be affected by the choice of $\gamma$, only the runtime. Further let $T=K_1/\gamma^2$, where $K_1 = 16/3$, and assume $\delta < 0.1$. The algorithm then produces $X_0 = x_0, X_1,...,X_T \in \mathbb{R}^n$ according to the following algorithm\\
\begin{algorithm}[H]
\SetAlgoLined
% \KwResult{Write here the result }
 
 \For{$t = 1,...,T$}{
  $\cdot$ Let $C_t^{var} \gets C_t^{var}(X_{t-1}) = \{i \in [n]: |(X_{t-1})_i| \geq 1-\delta\}$ be the set of variable constraints ``nearly hit" on the previous iteration\;
  $\cdot$ Let $C_t^{\disc} \gets C_t^{\disc}(X_{t-1}) = \{j \in [m]: |\langle X_{t-1}-x_0,v_j\rangle| \geq c_j-\delta\}$ be the set of discrepancy constraints "nearly hit" on the previous iteration\;
  $\cdot$ Let $\mathcal{V}_t \gets \mathcal{V}(X_{t-1}) = \{u \in \mathbb{R}^n: u_i=0 \forall i \in C_t^{var}, \quad \langle u, v_j\rangle = 0 \forall j \in C_t^{\disc}\}$ be the subspace orthogonal to the ``nearly hit" variable and discrepancy constraints.\;
  $\cdot$ Set $X_t \gets X_{t-1} + \gamma U_t$, where $U_t \sim \mathcal{N}(\mathcal{V}_t)$
 }
 \caption{The Brownian Motion process for Theorem \ref{Coloring}}
\end{algorithm}

When we say that $U \sim \mathcal{N}(\mathcal{V}_t)$, we are referring to the standard multi-dimensional Gaussian distribution: $U = U_1v_1 + ... + U_dv_d$ where $\{v_1,..,v_d\}$ is an orthonormal basis for $\mathcal{V}_t$ and $U_1, ...,U_d \sim \mathcal{N}(0,1)$ are all independent. 

\newpage

\begin{center}
    4.2. Analysis Outline.
\end{center}

We seek to prove the following:
\begin{lemma}\label{coloring-analysis}
We have that Theorem \ref{Coloring} holds for $X_T$ in the above algorithm, and that with probability at least $0.1$, $X_0,...,X_T \in \mathcal{P}$.
\end{lemma}

We begin with a useful claim regarding the behavior of the random walk.

\begin{claim}\label{Ortho}
For all $t$ we have that $C_t^{var} \subseteq C_{t+1}^{var}$ and similarly $C_t^{\disc} \subseteq C_{t+1}^{\disc}$ for $t=0,...,T-1$. This further implies that $\dim(\mathcal{V}_t) \geq \dim(\mathcal{V}_{t+1})$.
\end{claim}
\begin{proof}
Intuitively, we are taking Gaussian steps orthogonal to the subspace $C_t$, so at each step we should never be able to remove any elements in $C_t^{var}$ or $C_t^{\disc}$. Formally, let $i \in C_t^{var}$. Then $U_t \in \mathcal{V}_t$ which implies $(U_t)_i = 0$. This implies that $(X_t)_i = (X_{t-1})_i$ and $i \in C_{t+1}^{var}$ as desired. The argument is very similar for the discrepancy constraints.
\end{proof}

Now, we can begin to look at the results of the algorithm. 

First, we can prove that with good probability, our Brownian motion will not leave the polytope. The ``nearly hit" constraints serve this purpose; we select step size $\gamma$ small enough that whenever a solution approaches a constraint, it is more likely to fall into the $\delta$-band of the constraint than it is to break the polytope. Once it falls into this band, Claim \ref{Ortho} implies that it will never break the polytope. This can be shown formally using Gaussian tailbounds.

Next, we argue that the algorithm satisfies many variable constraints and few discrepancy constraints with high probability. Using our bound on the discrepancy coefficients as well as Gaussian tailbounds, we can demonstrate that the easily-satisfiable discrepancy constraints is small, and that it is unlikely that many other ones are met. With this in mind, at any time $t$ note that there are two scenarios for $C_t^{var}$: if it is large then we are done, and if it is small then our Brownian motion is less constrained so we expect to take steps of larger magnitude. Thus, we argue that by time $T$ it is likely that we ``nearly-hit" many variable constraints.

Finally, we look at the computational complexity of the algorithm, which is claimed to be $O((n+m)^3\delta^{-2}\log(nm/\delta))$. The paper does not provide a full justification of this runtime, but we believe it to be inaccurate. 

Computing $C_t^{var}$ and $C_t^{\disc}$ given $X_{t-1}$ takes $O(nm)$ time, since computing $C_t^{\disc}$ requires the computation of $m$ dot products in $\mathbb{R}^n$. We can sample from $\mathcal{N}(\mathcal{V}_t)$ by constructing an orthonormal basis for $\mathcal{V}_t$. We do this by constructing an orthogonal basis using our constraints, and using the completion theorem to find a basis of $\mathcal{V}_t$. Finding a basis from $n+m$ constraint vectors requires Gaussian elimination, so it takes $O((n+m)^3)$ time. 

Now, we have to repeat this for $T$ rounds, so the runtime should be expressible as $O((n+m)^3 T)$. Note that $T = O(1/\gamma^2)$, so the runtime described by \cite{Lovett} holds in the case where 
\begin{align*}
    \frac{1}{\gamma^2} &= O(\delta^{-2}\log(nm/\delta)) \\
    \frac{1}{\gamma} &= O\left(\frac{1}{\delta}\sqrt{\log(nm/\delta)}\right) \\
    \delta &= O(\gamma\sqrt{\log(nm/\delta)}).
\end{align*}
However, in the paper, $\gamma$ is selected under the condition $\delta = O(\gamma\sqrt{nm/\gamma})$. Of course, this ends up being a small distinction for $nm \gg \delta$, but it is still worth noting.

A full proof is provided in the appendix at section \ref{full-proof}. 

% Dalton is very upset that Pancake Bread exists. Why??? its incredible.\\

\section{The Discrepancy Minimizer}
For the purposes of brevity, we only provide a proof for Theorem \ref{result}.

\begin{proof}[Proof (Theorem \ref{result})] To find our full coloring, we will simply repeatedly use Theorem \ref{Coloring}. For $m$ sets on a universe of size $n$, we'll select $\delta = 1/(8 \log m)$ and $c_1, ..., c_m = 8\sqrt{\log(m/n)}$, and denote by $v_i ... v_m$ the indicator vectors for the sets. We'll use the partial coloring algorithm starting with vector $\vec{x}_0 = 0^n$ to find some vector $\vec{x}_1$ where $|\langle v_j, x_1 \rangle| < \sqrt{n}(8\sqrt{\log(m/n)})$ for all $j$ and where more than half of the points have values within the ``nearly-hit" bound. By Theorem \ref{Coloring}, this has probability of at least $0.1$, which we can boost by repeating as needed. 

Applying this iteratively to the vectors that haven't yet been assigned a partial coloring, we find that within $t = O(\log n)$ iterations every value in $x$ will be within $\delta$ of an assignment. When this occurs, for any $j \in [m]$, we note that $n_i < \frac{n}{2^i}$, so we have 
\begin{align*}
    |\langle v_j, x \rangle| &< \sum_{i=0}^t |\langle v_j, x_t \rangle| \\
    &< \sum \sqrt{n_i}8\sqrt{\log(m/n_i)} \\
    &< 8\sqrt{n} \sum_{i=1}^{\infty} \sqrt{\frac{i + \log(m/n)}{2^i}} \\
    &< C\sqrt{n\log(m/n)}
\end{align*}
for some constant $C$.

We then use this candidate solution and round it to an actual coloring. Knowing that each variable is within $\delta$ of either $1$ or $-1$, we'll set each variable to the one it is closer to with probability $(1+|x_i|)/2$, which means that $\mathbb{E}[\chi_i] = x_i$. Denoting $Y := \chi - x$ we have that the discrepancy for any set $j$ follows  
\[ |\langle \chi, v_j \rangle| \leq |\langle x, v_j \rangle| + |\langle Y, v_j \rangle|\]
due to triangle inequality. What's left, then, is to find an upper bound for $|\langle Y, v_j \rangle|$. Noting that $|Y_i| \leq 2$, $\mathbb{E}[Y_i]=0$, $\sigma^2(Y_i) \leq \delta$ (which the paper claims but we are only able to show this is true for $2\delta$), and $||v_j||_2 \leq \sqrt{n}$, the fact that $||v_j||_{\infty} \leq 1$ allows us to use a Chernoff bound and get 
\begin{align*}
    \Pr[|\langle Y, v_j \rangle| > 2 \sqrt{2\log m} \sqrt{n\delta}] &\leq 2\exp(-2\log m) \\
    &\leq 2/m^2 \\
    &\leq 1/2m
\end{align*}
for $m > 2$. Note that $\delta = 1/(8 \log m)$, so $2 \sqrt{2\log m} \sqrt{n\delta} = \sqrt{n}$, which means that across all $j$ we have $\Pr[|\langle Y, v_j \rangle| > \sqrt{n}] < 1/2$. Therefore, with probability at least $1/2$, $\disc(\chi) \leq C\sqrt{n\log(m/n)} + \sqrt{n} < K\sqrt{n\log(m/n)}$, as desired.
\end{proof}

\newpage

\section*{References}
\printbibliography[heading = none]

\newpage

\section{Appendix}

\subsection{A proof of the Beck-Fiala Theorem}\label{Beck-Fiala}
We present a proof of the Beck-Fiala theorem using only arguments from linear algebra.

\begin{proof} \cite{Chazelle}
We start by initializing all $\chi(i) = 0, \forall i \in [n]$, and we call all of these variables \textit{undecided}. We also call a set \textit{stable} if it has less than or equal to $t$ undecided elements. We also note that due to the constraint, there must be less $n$ sets that contain strictly more than $t$ elements to start off with (all of which are undecided upon initialization). If we impose the constraints that all of the elements in each unstable set must be zero, we get a system of less than $n$ equations, and $n$ variables. This tells us that there is at least one nontrivial solution to the system of equations, that changes only undecided variables and maintains that the discrepancy of all unstable sets remains zero. We can normalize this solution until at least one of the undecided variables is $\pm 1$. Then, this variable is decided, and we have a partial coloring. We now have at most $n-1$ undecided variables, and each undecided variable is in $(-1, 1)$. By the same argument from above, we have that the number of unstable sets is strictly less than the number of undecided variables, so we can repeat the procedure to find another nontrivial solution to our system of equations. We continue in the manner until all the sets are stable. Then we note that until a set is declared stable, its discrepancy is 0. Then, when it is declared stable, it has at most $t$ undecided variables, all of which are in $(-1,1)$. Then the process of deciding those variables changes the discrepancy of the set by strictly less than $2t$. And since the final discrepancy must be integral, we get the result.  
\end{proof}

\subsection{A Full Proof of Lemma \ref{coloring-analysis}} \label{full-proof}
We have already argued about the runtime of the algorithm. Here, we must show that the solution is unlikely to leave the polytope, that few discrepancy constraints are met, and that many variable constraints are met.

\begin{claim}\label{Polytope}
For $\gamma \leq \delta / \sqrt{c\log (mn/\gamma)}$ and $c$ a sufficiently large constant, with probability at least $1-1/(mn)^{c-2}$ we have that $X_0,...,X_T \in \mathcal{P}$
\end{claim}

To prove the above claim, we will need to use a Gaussian tail bound:
\begin{claim} \label{tailbound}
For any $\lambda >0$, $P(|G| \geq \lambda ) \leq 2\exp(-\lambda^2/2)$, where $G \sim \mathcal{N}(0,1)$
\end{claim}

\begin{proof}
We have that 
$$P(|G|>\lambda) = 2P(G > \lambda) = 2 \int_{\lambda}^{\infty} \frac{1}{\sqrt{2\pi}}\exp(-t^2/2)dt \leq 2 \int_{\lambda}^{\infty} \frac{t}{\lambda}\frac{1}{\sqrt{2\pi}}\exp(-t^2/2)dt  = \frac{2\exp(-\lambda^2/2)}{\sqrt{2\pi}\lambda}$$
From here, it is easy to see that for $\lambda \geq 1/\sqrt{2\pi}$ we have that $\frac{2\exp(-\lambda^2/2)}{\sqrt{2\pi}\lambda} \leq 2\exp(-\lambda^2/2) $ as desired. For the case when $\lambda \leq 1/\sqrt{2\pi}$, it is easy to see that $2\exp(-\lambda^2/2) > 1$ so the bound is trivial.
\end{proof}

\begin{proof}[Proof of Claim \ref{Polytope}]
Clearly $X_0 = x_0 \in \mathcal{P}$. We further define $E_t := \{X_t \not\in \mathcal{P} | X_0,...,X_{t-1} \in \mathcal{P}\}$ denote the event that $X_t$ is the first element of the sequence not in $\mathcal{P}$. \ We then have
$$Pr(X_0,...,X_T \in \mathcal{P}) = 1-\sum_{t=1}^T Pr(E_t).$$

The next step of the proof is clearly to calculate $Pr(E_t)$. In order for $E_t$ to happen, it must be the case that either a variable constraint or a discrepancy constraint was violated. Lets first look at the variable constraint case: say $(X_t)_i>1$. Since $X_{t-1} \in \mathcal{P}$, we must have that $(X_{t-1})_i \leq 1$. Yet, if $(X_{t-1})_i \geq 1-\delta,$ then $i \in C_t^{var}$ so $(X_{t-1})_i = (X_{t})_i$. Thus, for the constraint to be violated, we must have had that $(X_{t-1})_i < 1-\delta $. Then, in order for $(X_t)_i$ to be greater than 1, and since $X_t = X_{t-1} + \gamma U_t$, we must have that $|U(t)_i| \geq \delta/\gamma$.

Now, let's look what must happen in order for $X_t$ to violate a discrepancy constraint. First we define $W := \{e_1, ..., e_n, v_1,...,v_m\}$. By our construction of $W$, we conclude that if $E_t$ holds then we must have that $|\langle X_t - X_{t-1} ,w\rangle| \geq \delta$ for some $w \in W$. This is equivalent to saying that $|\langle U_t, w \rangle| \geq \delta/\gamma$ for that same $w$. We note here that, once again by construction of $W$, we have that if  $|U(t)_i| \geq \delta/\gamma$ holds, then we must have that $|\langle U_t, w \rangle| \geq \delta/\gamma$ holds for some $w$, in particular it holds if we pick $w=e_i$. However, the reverse does not hold. Therefore, it since the event of a variable constraint being violated is entirely contained in the event of a discrepancy constraint being violated, it suffices to bound $Pr[|\langle U_t, w \rangle| \geq \delta/\gamma]$. In order to bound this, we need the following claim:

\begin{claim} \label{subspace}
Let $V \subseteq \mathbb{R}^n$ be a subspace and let $G \sim \mathcal{N}(V)$. Then for all $u \in \mathbb{R}^n$, $\langle G, u \rangle \sim \mathcal{N}(0, \sigma^2)$, where $\sigma^2 \leq \|u\|_2^2$
\end{claim}

\begin{proof}
We have that $G = G_1v_1 + ... + G_dv_d$, where $\{v_1,...,v_d\}$ is an arbitrary orthonormal basis for $V$ and $G_1, ..., G_d$ are all standard normals and independent. Then $\langle G, u \rangle = \sum_{i=1}^d \langle v_i, u \rangle G_i$. This is a Gaussian RV which has mean zero, and variance $\sum_{i=1}^d \langle v_i, u \rangle^2$. This equation is simply equal to $\|\proj_V u\|_2^2$, the norm squared of the projection of $u$ onto the span of $V$. Therefore, we have that $\sum_{i=1}^d \langle v_i, u \rangle^2 \leq \|u\|_2^2$, and we are done.
\end{proof}
Now we can use the above claim, and we have that $\langle U_t, w\rangle$ is Gaussian with mean 0 and variance at most 1. Then, by claim \ref{tailbound}, we have that:
$$Pr[|\langle U_t, w \rangle| \geq \delta/\gamma] \leq 2\exp(-(\delta/\gamma)^2/2) $$
Now, by our choices of variables, we have that $\delta/\gamma = \sqrt{C\log (nm/\gamma})$ and $T = O(1/\gamma^2)$. Therefore, we have
$$Pr[X_0 \cup ... \cup X_t \not\in \mathcal{P}] = \sum_{t-1}^T Pr[E_t]$$
which, by a union bound, 
$$\leq \sum_{t-1}^T \sum_{w \in W} Pr[|\langle U_t, w \rangle| \geq \delta/\gamma] \leq T (n+m)\cdot 2\exp(-(\sqrt{C\log (nm/\gamma}))^2/2) = T(n+m)\cdot 2\left (\frac{\gamma}{nm}\right )^C $$
$$\leq T(nm)\frac{\gamma^2}{(nm)^C} \leq \frac{1}{(nm)^{C-2}}$$
For large enough C, since we have that $\gamma <1$ and $(nm) >1$.
\end{proof}
We are now well on our way to proving Lemma \ref{coloring-analysis}. The intuition behind the remaining steps is as follows. We will use the constraint on our discrepancy thresholds $c_j, j \in [m]$ to argue first that $\mathbb{E}[|C_T^{var}|] \ll n$. This will be useful because it will mean that $\dim (\mathcal{V}_{t-1})$ will be larger, which means that $\mathbb{E}[\|X_t\|^2]$ should increase more appreciably compared to the previous timestep. At any given timestep ,either $|C_t^{var}|$ is large and we are done, or $|C_t^{var}|$ is small and once again $\dim (\mathcal{V}_{t-1})$ and we will be taking bigger steps. We note also that in order to prove the lemma, we really only need to show that $\mathbb{E}[|C_t^{var}|] = \Omega(n)$, since if we achieve this, then we can use this along with the fact that $|C_T^{var}|$ is upper bounded by $n$ to show that $Pr[|C_t^{var}| < n/2] < 0.9$.\\
We first show that $\mathbb{E}[|C_t^{\disc}|]$ is small; that is, on average very few discrepancy constraints are ever nearly hit.

\begin{claim}
$\mathbb{E}[|C_t^{\disc}|] < n/4$
\end{claim}
\begin{proof}
We let $J := \{j: c_j \leq 10\delta \}$. In order to bound the size of $J$, we have that from our constraints
$$n/16 \geq \sum_{j \in J} \exp (-c_j^2/16) \geq |J| \cdot \exp(-100\delta^2/16) \geq |J| \cdot \exp(-1/16) > 9|J|/10$$
Since $\delta <0.1$. So we have then that $|J| \leq 1.2n/16 < 2n/16$. Now we consider the $j \not \in J$. If $j \in C_T^{\disc}$, then $|\langle X_T-x_0, v_j\rangle| \geq c_j - \delta \geq 0.9c_j$. We want to bound the probability that his occur. Via our update formula, we have that $X_T = x_0 + \gamma(U_1+...+U_T)$.  We then define $Y_i = \langle U_i, v_j \rangle$. We then have that for $j \not\in J$
$$Pr[j \in C_T^{\disc}] = Pr[|Y_1 + ... + Y_T| \geq 0.9c_j/\gamma]$$
We will also need the following Lemma:
\begin{lemma}[\cite{Bansal}]
Let $X_1,...X_T$ be random variables, and let $Y_1,...,Y_T$ be RVs where each $Y_i$ is a function of $X_i$. Suppose that for all $1 \leq i \leq T$, $Y_i|X_1, ..., X_{i-1}  $ is Gaussian with mean zero and variance at most one. Then for any $\lambda >0$:
$$Pr[|Y_1 + ...+ Y_T| \geq \lambda \sqrt{T}] \leq 2 \exp(-\lambda^2/2)$$
\end{lemma}
The proof of the above lemma is a generalization of the proof for Claim \ref{tailbound}, and is omitted. We note that $Y_i = \langle U_i, v_j \rangle \sim \mathcal{N}(0, \sigma^2)$, where $\sigma^2 \leq \|v_j\|^2 = 1$ by our assumption at the beginning of the problem that we had normalized all the the $v_i$. We can apply the above lemma to our $Y_i$'s, since $Y_i$ is a function of $U_i$ and $Y_i|U_1,...,U_{i-1}$ is Gaussian with mean zero and variance at most one to get that 
$$Pr[j \in C_T^{\disc}] \leq 2\exp(-(0.9c_j)^2/2\gamma^2T$$
which, since $T = K_1/\gamma^2$
$$ = 2\exp(-(0.9c_j)^2/2K_1) < 2\exp(-c_j^2/2)$$
We therefore have that 
$$\mathbb{E}[|C_T^{\disc}|] \leq |J| + \sum_{j \not\in J} Pr[j \in C_T^{\disc}] < 2n/16+2n/16 = n/4$$
Where above we have used conditional expectations, and assumed in the worst case that every element in $J$ is in $C_T^{\disc}$ and we have also used the constraint $\sum_{j=1}^m\exp(-c_j^2/16) \leq n/16$.
\end{proof}

\begin{claim} \label{bound-x-t}
$\mathbb{E}[\|X_T\|_2^2] \leq n$
\end{claim}
\begin{proof}
We start by noting that it suffices to show that $\mathbb{E}[(X_T)_i^2] \leq 1$ for all $i \in [n]$, since $\mathbb{E}[\|X_T\|_2^2] = \sum_i \mathbb{E}[(X_T)_i^2]$. We have that
$$\mathbb{E}[(X_T)_i^2] = \mathbb{E}[(X_T)_i^2|i \not \in C_T^{var}]Pr[i \not\in C_T^{var}] + \sum _{t=1}^T \mathbb{E}[(X_T)_i^2|i \in C_t^{var} \setminus C_{t-1}^{var}]Pr[i \in C_t^{var} \setminus C_{t-1}^{var}]$$
Now, we have that clearly $\mathbb{E}[(X_T)_i^2|i \not \in C_T^{var}] \leq 1$. For the rest of the terms, we have:
$$\mathbb{E}[(X_T)_i^2|i \in C_t^{var} \setminus C_{t-1}^{var}] = \mathbb{E}[(X_t)_i^2|i \in C_t^{var} \setminus C_{t-1}^{var}]$$
$$ = \mathbb{E}[((X_{t-1})_i+\gamma(U_t)_i)^2|i \in C_t^{var} \setminus C_{t-1}^{var}] \leq \mathbb{E}[(1-\delta+\gamma(U_t)_i)^2|i \in C_t^{var} \setminus C_{t-1}^{var}]$$
$$ = (1-\delta)^2 + 2(1-\delta)\gamma\mathbb{E}[(U_t)_i] + \gamma^2 \mathbb{E}[(U_t)_i^2]$$
Here, we note that $\mathbb{E}[(U_t)_i^2] = 1$ and $\mathbb{E}[(U_t)_i] = 0$, so we have
$$ = (1-\delta)^2 + \gamma^2 \leq 1-\delta + \gamma < 1$$
by our construction of $\gamma$.
\end{proof}

\begin{claim}
$\mathbb{E}[|C^{var}_T|] \geq 0.56n$.
\end{claim}
\begin{proof}
We will use the high average norm of $X_t$ and low number of discrepancy constraints broken to demonstrate that the number of variable constraints broken is high with high probability. Note that 
\begin{align*}
    \mathbb{E}[||X_t||^2_2] &= \mathbb{E}[||X_{t-1} + \gamma U_t||_2^2] \\
    &= \mathbb{E}[||X_{t-1}||^2_2] + 2\mathbb{E}[||X_{t-1} \cdot \gamma U_t||_2] + \mathbb{E}[||\gamma U_t||_2^2] \\
    &= \mathbb{E}[||X_{t-1}||^2_2] + \mathbb{E}[||\gamma U_t||_2^2] \\
    &= \mathbb{E}[||X_{t-1}||^2_2] + \gamma^2\mathbb{E}[\text{dim}(\mathcal{V}_t)].
\end{align*}
We use the fact that $\mathbb{E}[U_t]=0$, and we use Claim \ref{subspace} as well. Then, by Claim \ref{bound-x-t}, we have 
\begin{align*}
    n &\geq \mathbb{E}[||X_T||_2^2] \\
    n &\geq \gamma^2 \sum_{t=1}^T \mathbb{E}[\text{dim}(\mathcal{V}_t)] \\
    n &\geq \gamma^2 |T|\mathbb{E}[\text{dim}(\mathcal{V}_T)] \\
    n &\geq K_1 \mathbb{E}[n - |C_T^{var}| - |C_T^{\disc}|] \\
    n/K_1 &\geq \mathbb{E}[n] - \mathbb{E}[|C_T^{var}|] - \mathbb{E}[|C_T^{\disc}|] \\
    \mathbb{E}[|C_T^{var}|] &\geq n (1 - 1/K_1) - n/4 \\
    \mathbb{E}[|C_T^{var}|] &\geq n (1 - 3/16 - 1/4) \\
    \mathbb{E}[|C_T^{var}|] &\geq 0.5625n.
\end{align*}
\end{proof}

Now we can fully prove Lemma \ref{coloring-analysis}. From Claim 14 and the fact that $|C_T^{var}| \leq n$, in the worst case we have that with probability $0.88$, $|C_T^{var}| = n/2$, and with probability $0.12$, $|C_T^{var}| = n$. This maximizes the number of instances where $|C_T^{var}| \leq n/2$, while still maintaining the fact that $\mathbb{E}[|C_T^{var}|] \geq 0.56$. Therefore we have that $Pr[|C_T^{var}| > n/2] \geq 0.12$. Combining this with claim 8 tells us that with probability at least $0.12 - 1/poly(m,n)$ we achieve the partial coloring, and we are done.

\end{document}